\documentclass[twocolumn,superscriptaddress,english,pre,showpacs,longbibliography]{revtex4-2}
\usepackage{hyperref}
\usepackage{graphicx}
\usepackage{dcolumn}
\usepackage{bm}
\usepackage{subfigure}%

\usepackage{url}
\usepackage{amsthm,amsmath,amssymb}
\usepackage{mathrsfs}

\newtheorem{corollary}{Corollary}

\begin{document}


\title{Message Passing Variational Autoregressive Network\\ for Solving Intractable Ising Models}

\author{Qunlong Ma}
\affiliation{
Henan Key Laboratory of Network Cryptography Technology, Zhengzhou 450001, China}

\author{Zhi Ma}
\affiliation{
Henan Key Laboratory of Network Cryptography Technology, Zhengzhou 450001, China}

\author{Jinlong Xu}
\affiliation{
Henan Key Laboratory of Network Cryptography Technology, Zhengzhou 450001, China}

\author{Hairui Zhang}
\affiliation{
Department of Algorithm, TuringQ Co., Ltd., Shanghai 200240, China}

\author{Ming Gao}
\email{gaoming@nudt.edu.cn}
\affiliation{
Henan Key Laboratory of Network Cryptography Technology, Zhengzhou 450001, China}

\date{\today}
\begin{abstract}

Many deep neural networks have been used to solve Ising models, including autoregressive neural networks, convolutional neural networks, recurrent neural networks, and graph neural networks. Learning a probability distribution of energy configuration or finding the ground states of a disordered, fully connected Ising model is essential for statistical mechanics and NP-hard problems. Despite tremendous efforts, a neural network architecture with the ability to high-accurately solve these fully connected and extremely intractable problems on larger systems is still lacking. Here we propose a variational autoregressive architecture with a message passing mechanism, which can effectively utilize the interactions between spin variables. The new network trained under an annealing framework outperforms existing methods in solving several prototypical Ising spin Hamiltonians, especially for larger spin systems at low temperatures. The advantages also come from the great mitigation of mode collapse during the training process of deep neural networks. Considering these extremely difficult problems to be solved, our method extends the current computational limits of unsupervised neural networks to solve combinatorial optimization problems.

\end{abstract}
\maketitle

\section{INTRODUCTION}

Many deep neural networks have been used to solve Ising models \cite{Carleo2019, Akinori2023}, including autoregressive neural networks \cite{Wu2019VAN, Hibat2021vca, McNaughton2020, Marylou2022, Wu2021, Panfeng2021}, convolutional neural networks \cite{Pixel2016}, recurrent neural networks \cite{Hibat2020RNN, Hibat2021vca}, and graph neural networks \cite{Dai2017LC, Li2018COG, Gasse2019, Joshi2019AnEG, Speck2020LHS, Schuetz2021PIGNN, Schuetz2022}. The autoregressive neural networks model the distribution of high-dimensional vectors of discrete variables to learn the target Boltzmann distribution \cite{Larochelle2011, Gregor2014, Germain2015, Uria2016NADE} and allow for directly sampling from the networks. However, recent works question the sampling ability of autoregressive models in highly frustrated systems, with the challenge resulting from mode collapse \cite{Inack2022, Ciarella_2023}. The convolutional neural networks \cite{Pixel2016} respect the lattice structure of the 2D Ising model and achieve good performance \cite{Wu2019VAN}, but cannot solve models defined on non-lattice structures. Variational classical annealing (VCA) \cite{Hibat2021vca} uses autoregressive models with recurrent neural networks (RNN) \cite{Hibat2020RNN} and outperforms traditional simulated annealing (SA) \cite{Kirkpatrick1983} in finding ground-state solutions of Ising problems. This advantage comes from the fact that RNN can capture long-range correlations between spin variables by establishing connections between RNN cells in the same layer. Those cells need to be computed sequentially, which results in a very inefficient computation of VCA. Thus, in a particularly difficult class of fully connected Ising models, Wishart planted ensemble (WPE) \cite{WPE2020}, Ref.~\cite{Hibat2021vca} only solves problem instances with up to 32 spin variables. Since a Hamiltonian with the Ising form can be directly viewed as a graph, it is intuitive to use graph neural networks (GNN) \cite{HAMMOND2011, NIPS2016_GCN2, semi2017GCN3} to solve it. While a GNN-based method \cite{Schuetz2021PIGNN} has been employed in combinatorial optimization problems with system sizes up to millions, which first encodes problems in Ising forms \cite{Andrew2014} and then relaxes discrete variables into continuous ones to use GNN, some researchers argue that this method does not perform as well as classical heuristic algorithms \cite{Boettcher2022, Angelini2022}. In fact, the maximum cut and maximum independent set problem instances with millions of variables used in Ref.~\cite{Schuetz2021PIGNN} are defined on very sparse graphs and are not hard to solve \cite{Angelini2022}. Also, a naive combination of graph convolutional networks (GCN) \cite{semi2017GCN3} and variational autoregressive networks (VAN) \cite{Wu2019VAN} (we denote it as 'GCon-VAN' in this work) is tried, but performs poorly in statistical mechanics problems defined on sparse graphs \cite{Panfeng2021}. Reinforcement learning has also been used to find the ground state of Ising models \cite{Mills2020FTG, Fan2023DIRAC}. In addition, recently developed Ising machines have been used to find the ground state of Ising models and have shown impressive performance \cite{Mohseni2022IsingMA}, especially those based on physics-inspired algorithms, such as SimCIM (simulated coherent Ising machine) \cite{Tiunov2019, King2018emulating} and simulated bifurcation (SB) \cite{Hayato2021, Oshiyama2022BenchmarkOQ}.

Exploration of new methods to tackle Ising problems of larger scale and denser connectivity is of great interest. For example, finding the ground states of Ising spin glasses on two-dimensional lattices can be exactly solved in polynomial time, while ones in three or higher dimensions is a non-deterministic polynomial-time (NP) hard problem \cite{FBarahona_1982}. Ising models correspond to some problems defined on graphs, such as the maximum independent set problems, whose difficulty in finding the ground state might depend on node's degree being larger than a certain value \cite{Barbier2013, CojaOghlan2015}. Design of neural networks to solve Ising models on denser graphs would lead to development of powerful optimization tools and further shed light on computational boundary of deep-learning-assisted algorithms.

Due to the correspondence between Ising models and graph problems, existing Ising-solving neural network methods can be described by the message passing neural networks (MPNN) framework \cite{Gilmer2017}. MPNN can be used to abstract commonalities between them and determine the most crucial implementation details, which help to design more complex and powerful network architectures. Therefore, we reformulate existing VAN-based network architectures into this framework and then explore more variants for designing new network architectures with meticulously designed message passing (MP) mechanisms to better address intractable Ising models. Here we propose a variational autoregressive architecture with a message passing mechanism and dub it message passing variational autoregressive network (MPVAN). It can effectively utilize the interactions between spin variables, including whether there are couplings and coupling values, while previous methods only considered the former, i.e., the correlations.

\begin{figure}[h]
\includegraphics[scale=0.55]{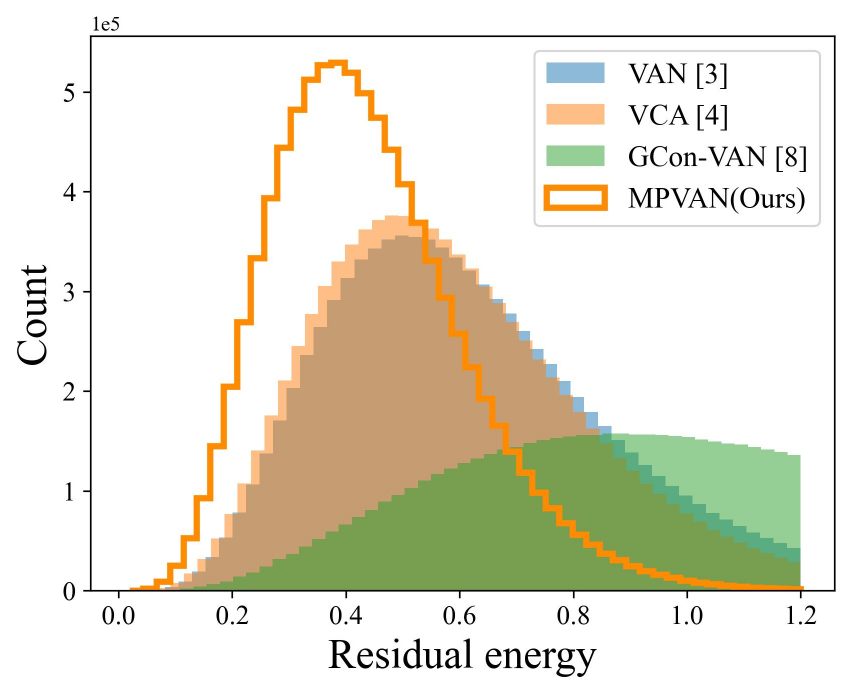}
\caption{\label{fig1} The residual energy histogram on the WPE with system size $N=60$ and difficulty parameter $\alpha=0.2$, which makes problem instances hard to solve. The residual energy is defined as the difference between the energy of the configurations sampled directly from the network after training and the energy of the ground state. Each method contains $9\times10^{6}$ configurations obtained from 30 instances and each for 30 runs.}
\end{figure}

We show the residual energy histogram on the Wishart planted ensemble (WPE) \cite{WPE2020} with a rough energy landscape in Fig.~\ref{fig1}. Compared to VAN \cite{Wu2019VAN}, VCA \cite{Hibat2021vca}, and GCon-VAN \cite{Panfeng2021}, the configurations sampled from MPVAN are concentrated on the regions with lower energy. Therefore, MPVAN has a higher probability of obtaining low-energy configurations, which is beneficial for finding the ground state, and it is also what combinatorial optimization is concerned about.

Numerical experiments show that our method outperforms existing methods in solving two classes of disordered, fully connected Ising models with extremely rough energy landscapes, the WPE \cite{WPE2020} and the Sherrington-Kirkpatrick (SK) model \cite{Sherrington1975SK}, including more accurately estimating the Boltzmann distribution and calculating lower free energy at low temperatures. The advantages also come from the great delay in the emergence of mode collapse during the training process of deep neural networks. Moreover, as the system size increases or the connectivity of graphs increases, MPVAN has greater advantages over existing methods in giving a lower upper bound to the energy of the ground state. Comparing to short-range Ising models such as the Edwards-Anderson model \cite{Edwards1975}, infinite-ranged interaction models (SK model and WPE) we considered are more challenging since there exist many loops of different lengths, which leads to more complicated frustrations. Considering these extremely difficult problems to be solved, our method extends the current computational limits of unsupervised neural networks \cite{Goodfellow2016} to solve intractable Ising models and combinatorial optimization problems \cite{Bengio2021}.

The paper is structured as follows. In Sec.~\ref{sec2}, we provide a detailed description of the message passing variational autoregressive network and provide a theoretical analysis. In Sec.~\ref{sec3}, we conduct experiments to benchmark our method and existing methods for solving intractable Ising models. We conclude and discuss in Sec.~\ref{sec4}.

\section{Message Passing Variational Autoregressive Network}
\label{sec2}

\begin{figure*}
\includegraphics[scale=0.063]{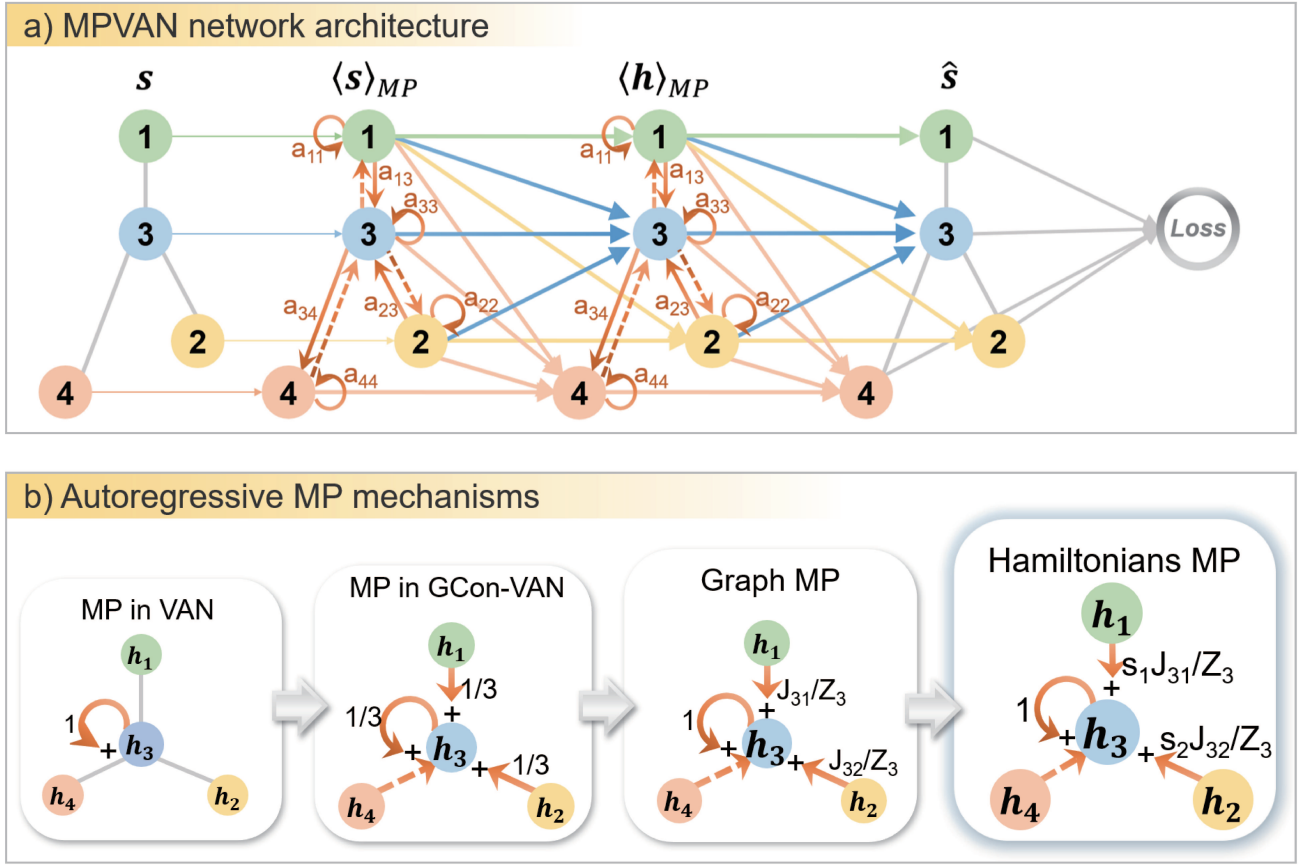}
\caption{\label{fig2} Schematic diagram of the network architecture of MPVAN and four autoregressive message passing mechanisms, which are shown on a problem instance with 3 edges and 4 spins. The spins are represented separately with numbers 1 to 4, and node features are represented separately with $h_i, i=1, 2, 3, 4$. (a) The network architecture of MPVAN. The spin configuration $\textbf{s}=\{\pm 1\}^N$ is the input to the network, $\hat{\textbf{s}}$ is the output, and $\textbf{h}$ denotes the hidden layer. The ${\langle \textbf{s} \rangle_{MP}}$ and ${\langle \textbf{h} \rangle_{MP}}$ are updated from $\textbf{s}$ and $\textbf{h}$ by message passing, respectively. The brown solid arrow indicates that neighboring nodes participate in message passing, while the brown dashed arrow indicates that there are connections between neighboring nodes but message passing is not performed to preserve the \textit{autoregressive\ property}. The $\{a_{ij}\}$ are the coefficient in message passing process, which vary for different message passing mechanisms. (b) The processes of four autoregressive message passing mechanisms when updating $h_{3}$. Under the MP mechanism used in VAN \cite{Wu2019VAN}, message passing is not performed, which is equivalent to the identity transformation of ${h_{3}}$. Under the MP mechanism used in GCon-VAN \cite{Panfeng2021}, message passing performs according to the adjacency matrix $A$, which updates the ${h_{3}}$ based on the topology structure of the graph. For the Graph MP mechanism we designed, message passing is performed by using the couplings $J_{ij}$ of the Hamiltonian, which updates $h_{3}$ based on the couplings and $Z_3=|J_{31}|+|J_{32}|$. The Hamiltonians MP mechanism we designed updates ${h_{3}}$ based on the couplings and values of neighboring spins $s_1$ and $s_2$, which is also the message passing mechanism used in MPVAN.}
\end{figure*}

The message passing variational autoregressive network (MPVAN) is to solve Ising models with the Hamiltonian as
\begin{equation}
    \label{eq1}
    H = -\sum_{\left<i,j\right>}{J_{ij}s_{i}s_{j}},
\end{equation}
where $\{s_{i}\}_{i=1}^{N} \in \{\pm 1\}^N$ are $N$ spin variables, and $\left<i,j\right>$ denotes that there is a non-zero coupling $J_{ij}$ between $s_i$ and $s_j$.

MPVAN is composed of an autoregressive message passing mechanism and a variational autoregressive network architecture, and its network architecture is shown in Fig.~\ref{fig2}(a). The input to MPVAN is configurations $\textbf{s}=\{s_{i}\}_{i=1}^{N}$ in a predetermined order of spins, and the $i_{th}$ component of the output, $\hat{s}_{i}$, means the conditional probability of $s_i$ taking $+1$ when given values of spins in front of it, $\textbf{s}_{<i}$, i.e., $\hat{s}_{i}=q_{\theta}(s_{i}=+1|\textbf{s}_{<i})$.

As in VAN \cite{Wu2019VAN}, the variational distribution of MPVAN is decomposed into product of conditional probabilities as
\begin{equation}
    \label{eq2}
    q_{\theta}(\textbf{s})=\prod^{N}_{i=1}{q_{\theta}}(s_i|s_1,s_2,\dots,s_{i-1}),
\end{equation}
where $q_{\theta}(\textbf{s})$ represents the variational joint probability and $q_{\theta} (s_i|s_1,s_2,\dots,s_{i-1})$ denotes the variational conditional probability, both of which are parametrized by trainable parameters $\theta$.

\subsection{MPVAN Layer}

MPVAN are constructed by stacking multiple \textit{message\ passing\ variational\ autoregressive\ network\ layers}. A single \textit{MPVAN layer} is composed of an autoregressive message passing process and nonlinear functions with trainable parameters defined by VAN \cite{Wu2019VAN}.

The input to the MPVAN layer is a set of node features, $\textbf{h}=\{\Vec{h}_1, \Vec{h}_2, \dots, \Vec{h}_N\}, \Vec{h}_i\in(0,1)^{F}$, where $F$ is the number of training samples. The layer produces a new set of node features, $\textbf{h}^{o}=\{\Vec{h}_1^{o}, \Vec{h}_2^{o}, \dots, \Vec{h}_N^{o}\}, \Vec{h}_i^{o}\in \mathbb{R}^F$, as the output. For brevity, we set $F=1$ and donate $\Vec{h}_i$ and $\Vec{h}_i^{o}$ as $h_i$ and $h_i^{o}$, respectively. The $\textbf{h}^{o}$ are obtained by
\begin{equation}
    \label{eq3}
\textbf{h}^{o}=\sigma(\langle \textbf{h}\rangle _{MP}W+b),
\end{equation}
where $sigmoid$ activation function $\sigma(x)=\frac{1}{1+e^{-x}}$ ranging in $(0,1)$ and thus $h_i^{o}\in(0,1)$. The $\langle \textbf{h}\rangle _{MP}=\{\langle h_1\rangle _{MP}, \langle h_2\rangle _{MP}, \dots, \langle h_N\rangle _{MP}\}$, $\langle h_i\rangle _{MP}\in \mathbb{R}$, denotes the updated node features from $\textbf{h}$ by \textit{message passing}. The $W$ and $b$ are layer-specific trainable parameters, and $W$ is a triangular matrix to ensure the autoregressive property \cite{Wu2019VAN,Gregor2014, Germain2015, Uria2016NADE}.

To show how to get $\langle \textbf{h}\rangle _{MP}$, we first review the message passing mechanism defined in the MPNN framework \cite{Gilmer2017}. We describe message passing operations on the current layer with node features $h_{i}$ and edge features $J_{ij}$. The message passing phase includes how to obtain neighboring messages $m_i$ and how to update node features $h_i$, which are defined as
\begin{equation}
\label{eq4}
\begin{aligned}
    &m_{i}=\sum_{j\in N_a(i)}M(h_{i},h_{j},J_{ij}),\\
    &\langle h_{i} \rangle _{MP}=U(h_{i},m_{i}),
\end{aligned}
\end{equation}
where $h_j$ is the node feature of $j$, and
\begin{equation}
\label{eq5}
    N_a(i)=\{j\ |\ j<i, J_{ij}\neq 0\},
\end{equation}
which denotes the neighbors located before node $i$. The $N_a(i)$ is used to preserve autoregressive property, which is different from general message passing mechanisms \cite{Gilmer2017} and graph neural networks \cite{semi2017GCN3, velickovic2018graph, Hamilton2017}. The message aggregation function $M(h_{i},h_{j},J_{ij})$ and node feature update function $U(h_{i},m_{i})$ are different across message passing mechanisms.

Now, existing VAN-based methods can be reformulated into combinations of VAN and different MP mechanisms. Then, we explore their variants and propose our method.

In VAN \cite{Wu2019VAN}, there is no message passing process from neighboring nodes. Therefore, the node features $h_{i}$ are updated according to
\begin{equation}
\label{eq6}
\begin{aligned}
    &m_{i} = 0,\\
    &\langle h_{i} \rangle _{MP} = h_{i},
\end{aligned}
\end{equation}
which is the first MP mechanism in Fig.~\ref{fig2}(b). Another successful variational autoregressive network approach is variational classical annealing (VCA) \cite{Hibat2021vca}, which uses an RNN architecture to take into account the correlation between hidden units in the same layer. The network structure of VCA is a special RNN architecture designed according to the topology of the model to be solved, thus taking into account the features of neighboring nodes through trainable parameters. Therefore, it is difficult to represent VCA as MPVAN with a special MP mechanism, since MP is free of trainable parameters.

In GCon-VAN \cite{Panfeng2021}, a combination of GCN \cite{semi2017GCN3} and VAN, the $\langle h_{i}\rangle_{MP}$ are obtained by
\begin{equation}
\label{eq7}
\begin{aligned}
    &m_{i} = \sum_{j\in N_a(i)} A_{ij} h_{j},\\
    &\langle h_{i}\rangle_{MP} = \frac{m_{i}+h_i}{deg(i)+1},
\end{aligned}
\end{equation}
where $A$ is the adjacency matrix of the graph, and $deg(i)$ represents the degree of node $i$. The $h_i$ is updated based on the connectivity of the graph, which is the second MP mechanism in Fig.~\ref{fig2}(b).

However, GCon-VAN performs poorly on sparse graphs in calculating physical quantities such as correlations and free energy \cite{Panfeng2021} and it performs even worse on dense graphs in our trial. It may be because GCon-VAN only considers connectivity and ignores the weights of neighboring node features. Also, from the results of VAN in Fig.~\ref{fig1}, the node feature $h_i$ itself should be highlighted rather than the small weight $\frac{1}{deg(i)+1}$ in Eq.~(\ref{eq7}).

Therefore, we explore more variants and propose three MP mechanisms. The $m_i$ are obtained by
\begin{subequations}
\label{eq8}
    \begin{align}
        & m_i = \sum_{j\in N_a(i)} |J_{ij}| h_{j} \label{eq8a},\\
        & m_i = \sum_{j\in N_a(i)} J_{ij} h_{j} \label{eq8b}, \\
        & m_i = \sum_{j\in N_a(i)} J_{ij} s_{j} h_{j}'\label{eq8c},
    \end{align}
\end{subequations}
where
\begin{equation}
\label{eq9}
h_{j}'=\frac{1+s_j}{2} h_{j}+ \frac{1-s_j}{2} (1-h_{j}).
\end{equation}
In above three mechanisms, the $\langle h_{i}\rangle_{MP}$ are obtained by
\begin{equation}
\label{eq10}
\langle h_{i} \rangle _{MP} = \frac{m_i}{\sum_{j\in N_a(i)}|J_{ij}|} + h_{i}.
\end{equation}

For the neighboring message $m_i$ in Eq.~(\ref{eq8a}), we consider the weight of neighboring node features instead of averagely passing those features in Eq.~(\ref{eq7}). Also, we increase the weight of $h_i$ in Eq.~(\ref{eq10}) for all mechanisms we designed. However, the Eq.~(\ref{eq8a}) ignores the influence of the sign of couplings $\{J_{ij}\}$.

Thus, we propose the MP mechanism made of Eq.~(\ref{eq8b}) and Eq.~(\ref{eq10}), which is the third MP mechanism in Fig.~\ref{fig2}(b). It uses the values of couplings $\{J_{ij}\}$ in the Hamiltonians to weight the neighboring node features. Since it is based on the graph defined by the target Hamiltonian, we dub it 'graph message passing mechanism' (Graph MP).

Previous methods and the above two MP mechanisms we designed do not make full use of the interactions between spin variables of the Hamiltonian in Eq.~(\ref{eq1}), and known values of $\textbf{s}_{<i}$ when updating the $h_i$. Intuitively, it may be helpful to consider those interactions and values of $\textbf{s}_{<i}$ in the message passing process.

Therefore, we propose the Hamiltonian MP mechanism composed of Eq.~(\ref{eq8c}) and Eq.~(\ref{eq10}), which is the fourth MP mechanism in Fig.~\ref{fig2}(b) and also the mechanism used in MPVAN. The $s_j=\pm 1$ is the known value of the neighboring spin $j$, and $h_j'\in(0,1)$ represents the probability of the spin $j$ taking $s_j$.

Based on the above definitions, it is reasonable to use $h_{j}'$ rather than $h_{j}$ in message passing. To illustrate, suppose for the neighboring spin $j$, $s_{j}=-1$ and $h_{j}=0.2$. Then, if we repeatedly sample spin $j$, we could obtain $s_{j}=-1$ with a probability of 0.8. It means that neighboring spin $j$ should have a greater impact on $h_i$ when it takes $-1$ than $+1$, but $h_{j}=0.2$ is not as good as $h_{j}'=0.8$ to reflect the great importance of $s_{j}=-1$. On the other hand, suppose for the neighboring spin $j$, $s_{j}=-1$ and $h_{j}=0.8$. We could obtain $s_{j}=-1$ with a probability of 0.2 if we sample spin $j$ again. At this time, using $h_{j}'$ instead of $h_{j}$ could show little importance of $s_{j}=-1$. So we think it makes sense to use $h_{j}'$ to reflect the effect of the spin $j$ taking $s_j$ in message passing.

Taking the graph with 4 nodes and 3 edges in Fig.~\ref{fig2}(b) as an example, applying the above four message passing mechanisms, the ${ \langle h_{3} \rangle }_{MP} $ are obtained as
\begin{subequations}
\label{eq11}
    \begin{align}
        &\langle h_{3}\rangle_{MP\ in\ VAN} =h_3, \label{11a}\\
        &\langle h_{3}\rangle_{MP\ in\ GCon-VAN}=\frac{h_{1}+h_{2}+h_{3}}{3}, \label{11b} \\
        &\langle h_{3}\rangle_{Graph\ MP}=\frac{J_{31}s_{1}h_{1}+J_{32}s_{2}h_{2}}{|J_{31}|+|J_{32}|} + h_{3}, \label{11c} \\
        &\langle h_{3}\rangle_{Hamiltonians\ MP}=\frac{J_{31}s_{1}h_{1}'+J_{32}s_{2}h_{2}'}{|J_{31}|+|J_{32}|} + h_{3}. \label{11d}
    \end{align}
\end{subequations}

We also consider more MP mechanisms and compare their performance in Appendix.~\ref{appen1}, where Hamiltonians MP always performs best. In addition, since an arbitrary message passing variational autoregressive network is constructed through stacking MPVAN layers, we also discuss the effect of the number of layers on the performance. MPVAN exhibits characteristics similar to GNN, i.e., there exists an optimal number of layers, which can be found in Appendix.~\ref{appen2}.

\subsection{Training MPVAN}

We then describe how to train MPVAN. In alignment with the variational approach employed in VAN, the variational free energy is used as loss function,
\begin{equation}
    \label{eq12}
    F_q=\sum_{\textbf{s}}{q_{\theta}(\textbf{s})\left[{E}(\textbf{s})+\frac{1}{\beta}\ln{q_{\theta}(\textbf{s})}\right]},
\end{equation}
where $\beta=1/T$ is inverse temperature, and $E(\textbf{s})$ is the $H$ of Eq.~(\ref{eq1}) related to a given configuration $\textbf{s}$. The configuration $\textbf{s}$ follows Boltzmann distribution $p(\textbf{s})={e^{-\beta E(\textbf{s})}}/Z$, where $Z=\sum_{\textbf{s}}{e^{-\beta E(\textbf{s})}}$. Since the KL divergence between the variational distribution $q_{\theta}$ and the Boltzmann distribution $p$ is defined as $D_{KL}(q_{\theta}||p)=\sum_{\textbf{s}}{q_{\theta}(\textbf{s})ln(\frac{q_{\theta}(\textbf{s})}{p(\textbf{s})})}=\beta(F_q-F)$ and is always non-negative, $F_q$ is the upper bound to the free energy $F=-(1/\beta)\ln Z$.

The gradient of $F_q$ with respect to the parameters $\theta$ is
\begin{equation}
\begin{aligned}
    \label{eq13}
\bigtriangledown_{\theta}F_q=\sum_{\textbf{s}}{q_{\theta}(\textbf{s})}{\left\{\left[{E}(\textbf{s})+\frac{1}{\beta}\ln{q_{\theta}(\textbf{s})})\right]\bigtriangledown_{\theta}\ln{q_{\theta}(\textbf{s})}\right\}}.
\end{aligned}
\end{equation}
With the computed gradients $\bigtriangledown_{\theta}F_q$, we iteratively adjust the parameters of the networks until the $F_q$ stops decreasing.

MPVAN is trained under an annealing framework, i.e., starting from the initial temperature $T_{initial}$ and gradually decreasing the temperature by annealing $N_{annealing}$ steps until the end temperature $T_{final}$. During each annealing step, we decrease the temperature and subsequently apply $N_{training}$ gradient-descent steps to update the network parameters, thereby minimizing the variational free energy $F_q$. As with the VAN \cite{Wu2019VAN}, the network is trained using the data produced by itself. After training, we can sample directly from the networks to calculate the upper bound to the free energy and other physical quantities such as entropy and correlations.

\subsection{Theoretical Analysis for MPVAN}

Compared to VAN\cite{Wu2019VAN}, MPVAN has an additional Hamiltonians message passing process. In this section, we will provide a theoretical and mathematical analysis of the advantages of the Hamiltonians message passing mechanism in MPVAN in Corollary.~\ref{cor1} below.

The goal of MPVAN is to be able to accurately estimate the Boltzmann distribution, i.e., configurations with low energy have a high probability and configurations with high energy have a low probability. Specifically, MPVAN is trained by minimizing the variational free energy $F_q$, composed of $\mathbb{E}_{\textbf{s}\sim q_{\theta}(\textbf{s})}{E}(\textbf{s})$ and $\mathbb{E}_{\textbf{s}\sim q_{\theta}(\textbf{s})}\ln{q_{\theta}(\textbf{s})}$.

\begin{corollary}
\label{cor1}

The Hamiltonians message passing process makes $\mathbb{E}_{\textbf{s}\sim q_{\theta}(\textbf{s})}{E}(\textbf{s})$ and $\mathbb{E}_{\textbf{s}\sim q_{\theta}(\textbf{s})}\ln{q_{\theta}(\textbf{s})}$ smaller, and therefore variational free energy $F_q$ smaller compared to no message passing.
\end{corollary}

\begin{proof}
We discuss message passing process that makes $\mathbb{E}_{\textbf{s}\sim q_{\theta}(\textbf{s})}{E}(\textbf{s})$ and $\mathbb{E}_{\textbf{s}\sim q_{\theta}(\textbf{s})}\ln{q_{\theta}(\textbf{s})}$ smaller separately.\\

\noindent \textbf{Step 1:} Making $\mathbb{E}_{\textbf{s}\sim q_{\theta}(\textbf{s})}{E}(\textbf{s})$ smaller.

When training MPVAN, it is impossible to exhaust all configurations to calculate the variational free energy $F_q$, so we use the mathematical expectation of training samples to estimate it. Therefore, we have
\begin{equation}
\label{eq14}
\mathbb{E}_{\textbf{s}\sim q_{\theta}(\textbf{s})}{E}(\textbf{s})= \frac{1}{N_s} \sum_{k=1}^{N_s}{E}(\textbf{s}_k),
\end{equation}
where $\textbf{s}_k$ is $k_{th}$ training samples and $N_s$ is the number of all training samples.

Consider the Hamiltonians message passing process (composed of Eq.~(\ref{eq8c}) and Eq.~(\ref{eq10})) for updating the $h_i$ as an example to analyze how message passing makes $\mathbb{E}_{\textbf{s}\sim q_{\theta}(\textbf{s})}{E}(\textbf{s})$ smaller. Since the message passing process maintains autoregressive property, making ${E}(\textbf{s})$ smaller for any configuration is equivalent to making the local Hamiltonian defined as $H_{local} = -\sum_{j<i}J_{ij}s_{i}s_{j}$ smaller, when given the values of its neighboring spins $\textbf{s}_{<i}$.

Since $H_{local}=-s_{i}\sum_{j<i}J_{ij}s_{j}$ with known $\sum_{j<i}J_{ij}s_{j}$, we are concerned about how the message passing affects the probability of $s_{i}$ taking $+1$ or $-1$. According to Eq.~(\ref{eq10}), the value of $\sum_{j<i}J_{ij}s_{j}h_{j}'$ plays an important role in that, and we discuss in 2 cases.

\textit{Case} 1: if $\sum_{j<i}J_{ij}s_{j}h_{j}'>0$, then according to Eq.~(\ref{eq10}), we have
\begin{equation}
\label{eq15}
\langle h_i\rangle_{MP}>h_{i},
\end{equation}
i.e., 
\begin{equation}
\label{eq16}
Pr[\langle s_{i}\rangle_{MP}=+1|\textbf{s}_{<i}]>Pr[s_i=+1|\textbf{s}_{<i}],
\end{equation}
where
\begin{subequations}
\label{eq17}
    \begin{align}
        &\langle s_{i}\rangle_{MP}=Bernoulli(\langle h_i\rangle_{MP}), \label{17a}\\
        &s_{i}=Bernoulli(h_i)\label{17b}.
    \end{align}
\end{subequations}
The $Bernoulli(p)$ denotes sampling from the Bernoulli distribution to output $+1$ with a probability $p$. Thus, we have
\begin{equation}
\label{eq18}
Pr[\langle H_{local}'\rangle_{MP}<0] > Pr[H_{local}'<0],
\end{equation}
where 
\begin{subequations}
\label{eq19}
    \begin{align}
        &\langle H_{local}'\rangle_{MP}=-\langle s_{i}\rangle_{MP}\sum_{j<i}J_{ij}s_{j}h_{j}',\label{19a}\\
        &H_{local}'= -s_{i}\sum_{j<i}J_{ij}s_{j}h_{j}',\label{19b}
    \end{align}
\end{subequations}
and $Pr[\cdots]$ denotes the probability of something.

\textit{Case} 2: if $\sum_{j<i}J_{ij}s_{j}h_{j}'<0$, then according to Eq.~(\ref{eq10}), we have
\begin{equation}
\label{eq20}
\langle h_i\rangle_{MP}<h_{i},
\end{equation}
i.e., 
\begin{equation}
\label{eq21}
Pr[\langle s_{i}\rangle_{MP}=-1|\textbf{s}_{<i}]>Pr[s_i=-1|\textbf{s}_{<i}],
\end{equation}
and also obtain the Eq.~(\ref{eq18}). Therefore, regardless of the value of $\sum_{j<i}J_{ij}s_{j}h_{j}'$, message passing process always makes $H_{local}'$ smaller compared with no message passing.

It can be found that the difference between $H_{local}$ and $H_{local}'$ is that the latter has $h_{j}'$. It encapsulates more information about the spin $j$ beyond the current configuration spin value $s_{j}$, which can indicate how much influence the message corresponding to the neighboring node features $h_j$ has on $\langle h_{i}\rangle_{MP}$ and, combining with the coupling $J_{ij}$, performs a weighted message passing. Thus, although not identical, it is possible to predict $H_{local}$ through $H_{local}'$ and thus we get the conclusion that Hamiltonians message passing mechanism makes local Hamiltonian smaller.\\

\noindent \textbf{Step 2:} Making $\mathbb{E}_{\textbf{s}\sim q_{\theta}(\textbf{s})}\ln{q_{\theta}(\textbf{s})}$ smaller.

Similar to $\mathbb{E}_{\textbf{s}\sim q_{\theta}(\textbf{s})}E(\textbf{s})$, we have
\begin{equation}
\label{eq22}
\mathbb{E}_{\textbf{s}\sim q_{\theta}(\textbf{s})}\ln{q_{\theta}(\textbf{s})}= \frac{1}{N_s} \sum_{k=1}^{N_s}\ln{q_{\theta}(\textbf{s}_k)}.
\end{equation}

The second derivative of $y(x)=ln(x)$ is $y(x)^{(2)}=-1/x^2$ and thus $y(x)$ is concave down, which satisfies $\frac{y(a)+y(b)}{2}<y(\frac{a+b}{2})$. From the analysis of \textbf{Step 1}, the message passing process improves (reduces) the probability of configurations with low (high) energy. Therefore, message passing makes $\sum_{k=1}^{N_s}\ln{q_{\theta}(\textbf{s}_k)}$ smaller compared to no message passing.
\end{proof}

In summary, combining the analyses of \textbf{Step 1} and \textbf{Step 2}, the message passing process makes the variational free energy $F_q$ smaller compared to no message passing.

\begin{figure}
\includegraphics[scale=0.5]{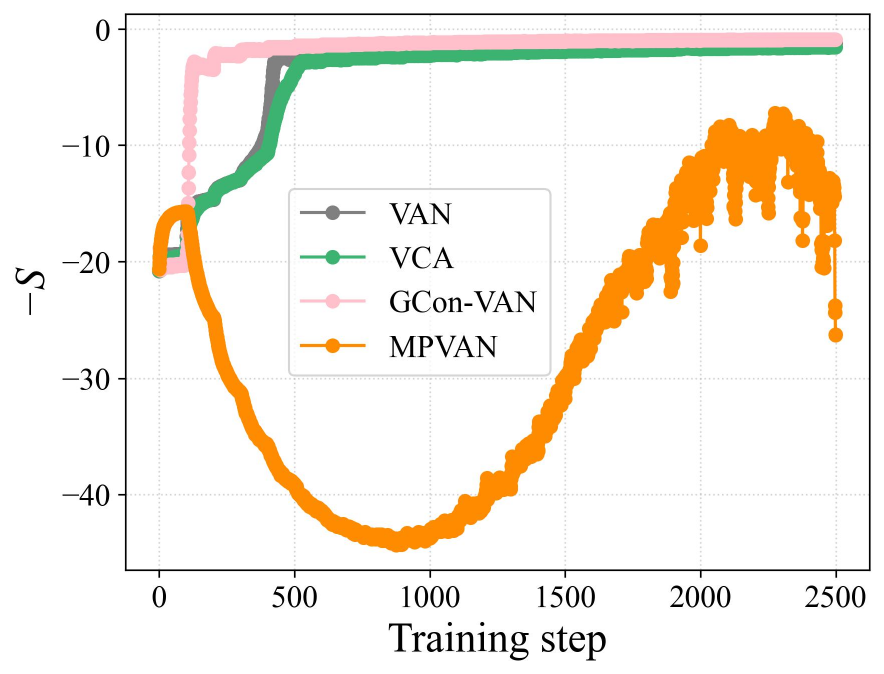}
\caption{\label{fig3} The negative entropy during training when $N_{annealing}=25$ and $N_{training}=100$, on the WPE with $N=30, \alpha=0.2$ and averaging on 10 runs.}
\end{figure}

\begin{figure*}
\includegraphics[scale=0.55]{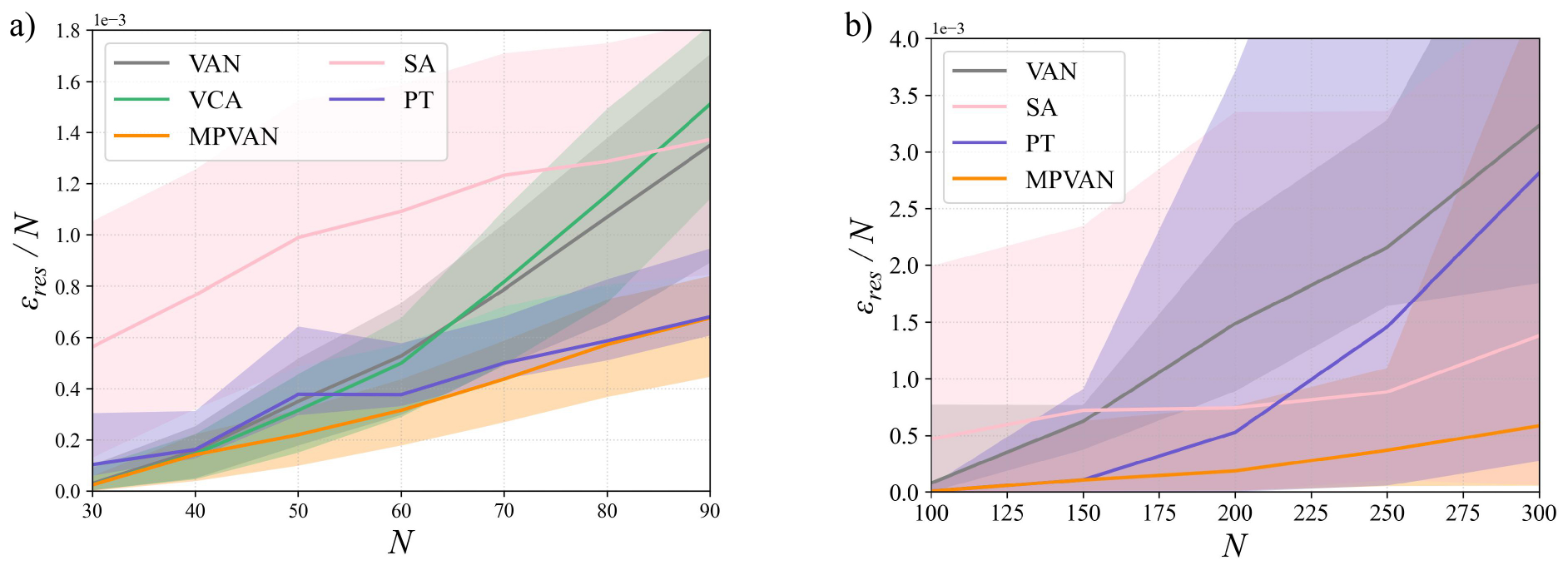}
\caption{\label{fig4} The residual energy per site of MPVAN with benchmark methods varies with system size $N$. (a) On the WPE, the $\epsilon_{res}/N$ averages on 30 instances and each for 30 runs, all instances with the system size $N$ and $\alpha=0.2$. When $N\geq 50$, the problem instances cannot be solved due to rough energy landscapes. (b) On the SK model, the residual energy per site averages on 30 instances and each for 10 runs. Since the energy of the ground state cannot be determined, we use the lowest energy across MPVAN, VAN, SA, and PT to replace it. Due to computational limitations, we exclude VCA from comparison when $N>100$ as its speed is about $N/10$ times slower than MPVAN when trained under the same hyperparameters. More details regarding computational speed of MPVAN and other methods can be found in Appendix~\ref{appen5}.}
\end{figure*}

\begin{figure}
\includegraphics[scale=0.55]{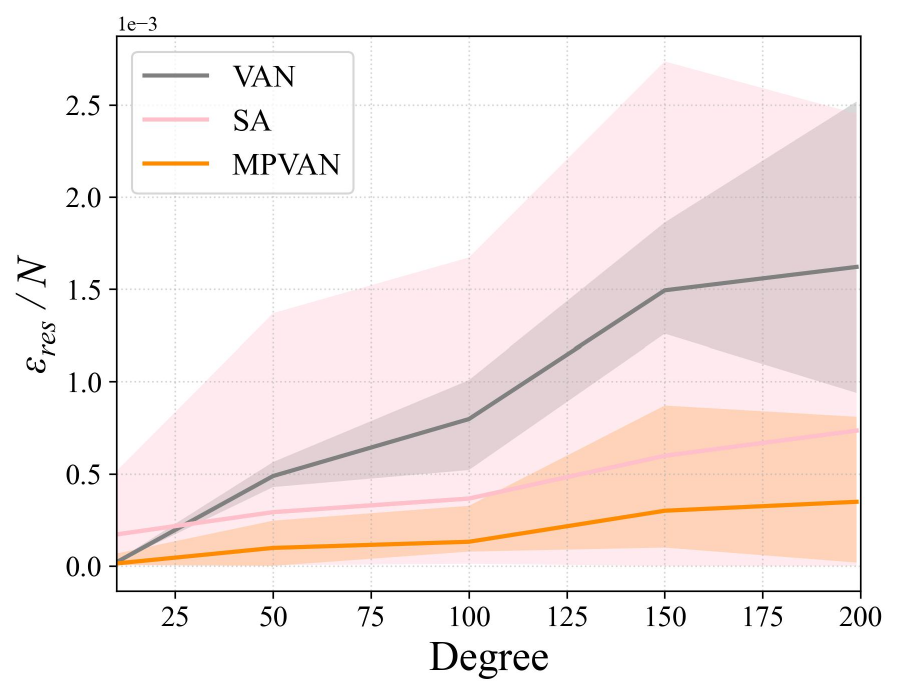}
\caption{\label{fig5} On the variants of the SK model, the residual energy per site of MPVAN with benchmark methods varies with average degree of each node in graphs with $N=200$ averaging on 30 randomly generated instances and each for 10 runs.}
\end{figure}

\section{NUMERICAL EXPERIMENTS}
\label{sec3}

As described in Sec.~\ref{sec2}, MPVAN also includes existing neural network approaches as special cases, with different message passing mechanisms. We conduct experiments in Appendix.~\ref{appen1} to compare MPVAN with multiple message passing mechanisms, where MPVAN with Hamiltonians MP performs best. Therefore below we consider comparing the MPVAN with Hamiltonians MP to existing methods, and if not specified, in the following MPVAN refers to for MPVAN with Hamiltonians MP. We experiment on two classes of fully connected and intractable models, the WPE and the SK models.

The Wishart Planted Ensemble (WPE) \cite{WPE2020} is a class of fully connected Ising models with an adjustable difficulty parameter $\alpha$ and planted solutions, which make it an ideal benchmark problem for evaluating heuristic algorithms. The Hamiltonian of the WPE is defined as
\begin{equation}
    \label{eq23}
    H = -\frac{1}{2}\sum_{i\neq j}{J_{ij}s_{i}s_{j}},
\end{equation}
where the coupling matrix $\{J_{ij}\}$ is a symmetric matrix that is subject to copula distribution. More details about the WPE can be found in Ref.~\cite{WPE2020}.

First, we discuss an important issue, mode collapse, which occurs when the target probability distribution has multiple peaks but networks only learn a few of them. It severely affects the sampling ability of autoregressive neural networks \cite{Ciarella_2023}. Entropy is commonly used in physics to measure the degree of chaos in a system. The greater the entropy, the more chaotic the system. Therefore, we can use the magnitude of entropy to reflect whether mode collapse occurs for the variational distribution. In our experiments, we investigate how the negative entropy, a part of variational free energy $F_q$ in Eq.~(\ref{eq12}), changes during training, which is defined as 
\begin{equation}
    \label{eq24}
        -S=\sum_{\textbf{s}}q_{\theta}(\textbf{s})\ln{q_{\theta}(\textbf{s})},
\end{equation}
where $S$ is entropy. Equivalently, the smaller the negative entropy, the more chaotic the system.

As shown in Fig.~\ref{fig3}, the change of $-S$ from MPVAN shows an increasing-decreasing-increasing trend, while $-S$ from other methods is monotonely increasing and quickly convergent. When $training\ step \geq\ 500$, mode collapse occurs for other methods, while until $training\ step \geq\ 2000$, it occurs for MPVAN. Therefore, MPVAN delays the emergence of mode collapse greatly. In addition, we also consider the impact of learning rates on the emergence of mode collapse in Appendix.~\ref{appen6}, where mode collapse always occurs later in MPVAN than in VAN.

\begin{figure*}
\includegraphics[scale=0.55]{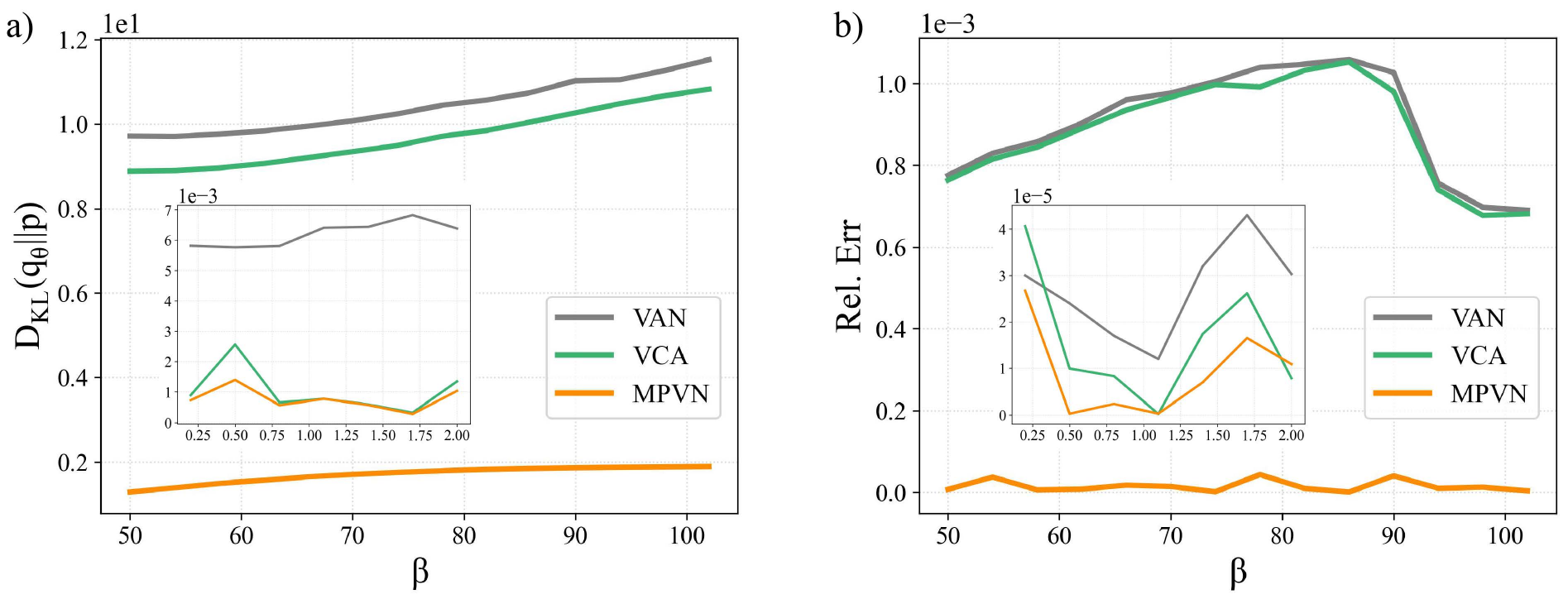}
\caption{\label{fig6} The KL divergence and relative errors vary with $\beta$ on the WPE with $N=20$ and $\alpha=0.05$. (a) The KL divergence $D_{KL}(q_\theta||p)$ between the variational distribution $q_\theta$ and the Boltzmann distribution $p$. The inset shows the $D_{KL}(q_\theta||p)$ when $\beta$ is small. (b) The relative errors of $F_q$ relative to exact free energy $F$. The inset shows the relative errors when $\beta$ is small.}
\end{figure*}

Next, we benchmark MPVAN with existing methods when calculating the upper bound to the energy of the ground state, i.e., finding a configuration $\textbf{s}$ to minimize the Hamiltonian in Eq.~(\ref{eq23}), which is also the core concern of combinatorial optimization. To facilitate a quantitative comparison, we employ the concept of residual energy, defined as
\begin{equation}
    \label{eq25}    \epsilon_{res}=\left[\left\langle{H_{min}}\right \rangle_{ava}-E_{G}\right]_{ava},
\end{equation}
where $H_{min}$ represents the minimum value of the Hamiltonians corresponding to $10^6$ configurations sampled directly from the network after training, and $E_{G}$ is the energy of the ground state. The $\left\langle{\dots}\right \rangle_{ava}$ denotes the average over 30 independent runs for one same instance, and $\left[\dots\right]_{ava}$ means averaging on 30 instances. In the figures representing the residual energy, such as Fig.~\ref{fig4} and Fig.~\ref{fig5}, the solid line in the figures indicates the average value of the residual energy, and the color band indicates the area between the maximum and minimum values of the residual energy of 30 independent runs for the corresponding algorithm. Both solid lines and color bands are obtained by averaging 30 instances.

As depicted in Fig.~\ref{fig4}(a), the residual energy obtained by our method consistently is lower than that of VAN, VCA, and SA across all system sizes when averaging on 30 instances and each for 30 runs. Even compared with state-of-the-art parallel tempering (PT) \cite{PT1986, PT2005}, MPVAN also exhibits slightly better performance in terms of average residual energy, but significantly better in terms of minimum residual energy. For WPE instances with system size $N=50$, MPVAN can still find the ground state with non-ignorable probability, but other methods cannot. As the system size increases, MPVAN has greater advantages over existing methods in giving a lower residual energy. Since the number of interactions between spin variables is $|\{J_{ij} \neq 0\}| = N^2$ for fully connected systems, larger systems have much more interactions between spin variables. The advantages indicate that using MPVAN to consider these interactions performs better in rougher energy landscapes. We always average 30 instances to reflect the general properties of models, and the differences between instances can be seen in Appendix.~\ref{appen1}. Also, each method runs independently 30 times on the same instances to weaken the influence of occasionality in the heuristics training, which can be found in Appendix.~\ref{appen3}.

Since these methods are trained in different ways, we keep the total number of training samples used in training and final sampling after training the same for all methods. The training samples of MPVAN consist of two parts, training samples and final sampling samples. Assuming annealing number $N_{annealing}$, training $N_{training}$ steps at each temperature, sampling $N_{trasam}$ samples each step in training, and final sampling $N_{finsam}$ samples, the total number of training samples for MPVAN is
\begin{equation}
    \label{eq26} N_{MPsam}=N_{annealing} \times N_{training} \times N_{trasam} + N_{finsam}.
\end{equation}
The training samples for VAN and VCA are the same as those for MPVAN, with some fine-tuning of parameters. Assuming the number of inner loops of SA is $N_{inlop}$ at each temperature, and the total number of training samples for SA is
\begin{equation}
    \label{eq27} N_{SAsam}=N_{annealing} \times N_{inlop}.
\end{equation}
Assuming the number of chains of PT is $N_{chain}$ and the number of random flips at each chain is $N_{rf}$, then the number of training samples for PT of 1 replica is
\begin{equation}
    \label{eq28} N_{PT1sam}=N_{chain} \times N_{rf}.
\end{equation}
Therefore, corresponding to 1 run of MPVAN, SA runs $\lceil N_{MPsam}/N_{SAsam} \rceil$ times independently and PT runs $\lceil N_{MPsam}/N_{PT1sam} \rceil$ replicas, and then they output their respective best results to benchmark MPVAN. It is important to note that the performance of each algorithm depends not only on the number of training samples used in training, but also on the training parameters, such as the annealing schedule and initial and final temperatures. We have fine-tuned the training parameters of each algorithm to maximize its best performance.

We also experiment on the Sherrington-Kirkpatrick (SK) model \cite{Sherrington1975SK}, which is one of the most famous fully connected spin glass models and has significant relevance in combinatorial optimization and machine learning applications \cite{Panchenko2012TheSM, Panchenko2013}. Its Hamiltonian is also in the form of Eq.~(\ref{eq23}), where $\{J_{ij}\}$ are from a Gaussian distribution with the variance $1/N$ and a symmetric matrix.

As illustrated in Fig.~\ref{fig4}(b), our method provides significantly lower residual energy than VAN, SA, and PT across all system sizes when averaging on 30 instances and each for 10 runs. Notably, as the system size increases, the advantages of our method over VAN, SA, and PT become even more pronounced, which is consistent with the trend observed in WPE experiments. We also show the approximating ratio results on the WPE and the SK model in Appendix.~\ref{appen4}, which is of concern to researchers of combinatorial optimization problems.

Inspired by the correlations between node's degree and difficulty in finding the ground state in maximum independent set problems \cite{Barbier2013, CojaOghlan2015}, in addition to experimenting on fully connected models, we also consider experiments on models with different connectivity, i.e., degrees of nodes in graphs. Since the SK model has been widely studied, it may be interesting to design new models based on the SK model. We generate models with different connectivity by deleting some couplings of the SK model and name them variants of the SK model. At each degree, we always randomly generate 30 instances, and the couplings $\{J_{ij}\}$ are from a Gaussian distribution with a variance of $1/N$ and a symmetric matrix. As shown in Fig.~\ref{fig5}, our method gives a lower residual energy than VAN and SA at all degrees. Moreover, as degree increases, the advantages of our method over VAN and SA become even more pronounced. The denser the graph, the larger the number of interactions between spin variables. The advantages show that our method, which takes into account these interactions, is able to give lower upper bounds to the free energy.

In the following, we focus on estimating the Boltzmann distribution and calculating the free energy as annealing. As a proof of concept, we use the WPE with a small system size of $N=20$, where $2^{N}$ configurations can be enumerated and the exact Boltzmann distribution and exact free energy $F$ can be calculated within an acceptable time. We set $\alpha=0.05$ and thus it is difficult to find the ground state due to strong low-energy degeneracy.

As shown in Fig.~\ref{fig6}, when the temperature is high, i.e., when $\beta$ is small, the $D_{KL}(q_\theta||p)$ and the relative errors of $F_q$ relative to exact free energy $F$ from MPVAN, VAN, and VCA are particularly small. Therefore, it is necessary to lower the temperature to distinguish them. As the temperature decreases, the probability of the configurations with low (high) energy in the Boltzmann distribution increases (decreases), thus making it more difficult for neural networks to estimate the Boltzmann distribution. However, we find that the $D_{KL}(q_\theta||p)$ obtained by our method is much smaller than that of VAN and VCA, which indicates that the variational distribution $q_\theta(\textbf{s})$ parameterized by our method is closer to the Boltzmann distribution. Similarly, our method gives a better estimation of free energy than VAN and VCA. These results illustrate that our method takes into account the interactions between spin variables through message passing and is more accurate in estimating the relevant physical quantities.

\section{CONCLUSION AND DISCUSSIONS}
\label{sec4}

In summary, we propose a variational autoregressive architecture with a message passing mechanism, which can effectively utilize the interactions between spin variables, to solve intractable Ising models. Numerical experiments show that our method outperforms existing methods including VAN, VCA, SA and even PT, in solving two prototypical Ising spin Hamiltonians, WPE and the SK model, including more accurately estimating the Boltzmann distribution and calculating lower free energy at low temperatures. The advantages also come from the great mitigation of mode collapse during the training process of deep neural networks. Moreover, as the system size increases or the connectivity of graphs increases, MPVAN has greater advantages over existing methods in giving a lower upper bound to the energy of the ground state.

Formally, MPVAN and GNN are similar. We notice that some researchers have recently argued that graph neural networks do not perform as well as classical heuristic algorithms on combinatorial optimization problems \cite{Boettcher2022, Angelini2022} for the method in Ref.~\cite{Schuetz2021PIGNN}. Our work, however, draws the opposite conclusion. We argue that when the problems are in rough energy landscapes and hard to find the ground state (e.g., WPE), our method performs significantly better than traditional heuristic algorithms such as SA and even slightly better than state-of-the-art PT. Our method is based on variational autoregressive networks, which are difficult to train due to slow speed when the systems are particularly large, and thus MPVAN is not easy to expand to very large problems. At the very least, we argue that MPVAN (or GNN) excels particularly well in certain intractable Ising models with rough energy landscapes, providing an alternative to traditional heuristics.


\begin{acknowledgments}
We thank Pan Zhang for helpful discussions on the manuscript. This work is supported by Project 61972413 (Z.M.) of the National Natural Science Foundation of China.\\
\end{acknowledgments}

\noindent\textbf{Author contributions}:

M.G. conceived and designed the project. Z.M. and M.G. managed the project. Q.M. and H.Z. performed all the numerical calculations and analyzed the results. G.M., Q.M., Z.M., and J.X. interpreted the results. Q.M. and G.M. wrote the paper.\\

\noindent\textbf{Data availability}:

The data that support the findings of this study are available from the corresponding author upon reasonable request.\\

\noindent\textbf{Code availability}:

The code that supports the findings of this study is available from the corresponding author upon reasonable request.\\

\noindent\textbf{Competing interests}:

The authors declare no competing interests.\\

\bibliography{aps-MPVAN}

\newpage \onecolumngrid \newpage { \center \bf \large  Supplemental Material for: \\ Message Passing Variational Autoregressive Network\\ for Solving Intractable Ising Models \vspace*{0.1cm}\\  \vspace*{0.0cm} } 

\begin{center} Qunlong Ma,$^{1}$ Zhi Ma,$^{1}$ Jinlong Xu,$^{1}$ Hairui Zhang,$^{2}$ and Ming Gao$^{1}$ \\ 
\vspace*{0.15cm}
\small{\textit{$^{1}$ Henan Key Laboratory of Network Cryptography Technology, Zhengzhou 450001, China}} \\
\small{\textit{$^{2}$ Department of Algorithm, TuringQ Co., Ltd., Shanghai 200240, China}} \\
\vspace*{0.25cm} 
\end{center}

\setcounter{section}{0}

\setcounter{equation}{0} 
\setcounter{figure}{0} 
\setcounter{page}{1} 
\makeatletter 
\renewcommand{\theequation}{S\arabic{equation}} 
\renewcommand{\thefigure}{S\arabic{figure}} 
\renewcommand{\bibnumfmt}[1]{[S#1]} 
\renewcommand{\citenumfont}[1]{#1} 

\section{More results about various message passing mechanisms}
\label{appen1}

In this section, we consider the combinations of various messaging mechanisms and VAN. Since our focus is on messaging mechanisms, we always use the messaging mechanism to refer to its combination with VAN instead.

First, we compare the three MP mechanisms we proposed in Eq.~(\ref{eq8a}), Eq.~(\ref{eq8b}), and Eq.~(\ref{eq8c}) on the same 30 instances in Fig.~\ref{fig1} of the main text. Since the Eq.~(\ref{eq8a}) used unsigned weights of graphs, we name it weighted message passing mechanism (Weighted MP).

\begin{figure}[h]
\includegraphics[scale=0.6]{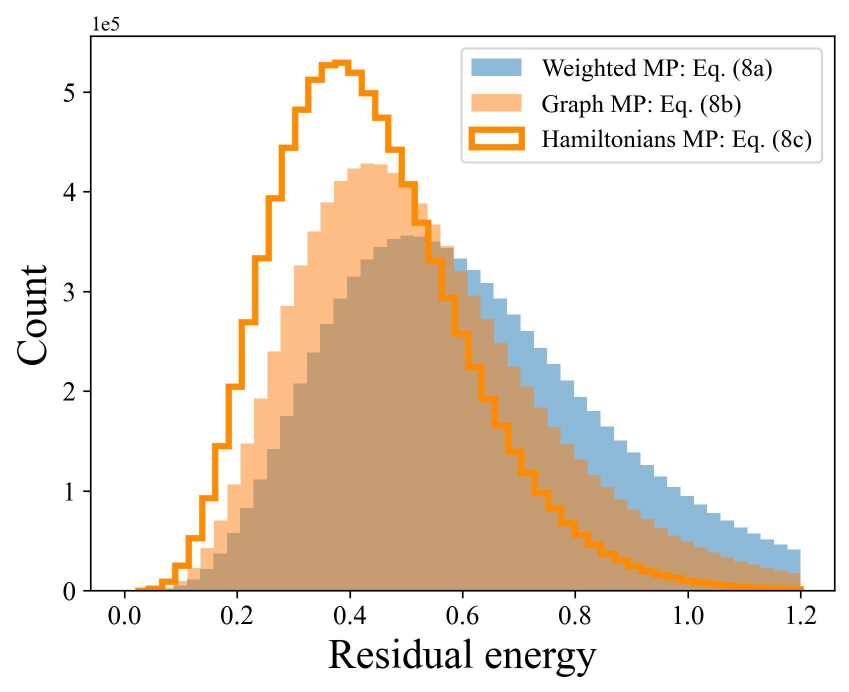}
\caption{\label{figS1} For the three message passing mechanisms we designed, the residual energy histogram on the WPE with system size $N=60$ and $\alpha=0.2$ same as in Fig.~\ref{fig1}. Each method contains $9\times10^{6}$ configurations obtained from 30 instances and each for 30 runs.}
\end{figure}

It can be found in Fig.~\ref{figS1} that the performance of Graph MP is better than Weighted MP, which illustrates the benefits of considering the sign of $\{J_{ij}\}$ in the message passing process. At the same time, Hamiltonians MP performs best among three mechanisms, and thus is the mechanism we ultimately utilize in MPVAN.

Second, we design experiments to determine that Hamiltonians MP performs best compared with other mechanisms, not solely as a result of considering the values of neighboring spins. Therefore, for MP in GCon-VAN (Eq.~(\ref{eq7}) in the main text) and Graph MP (Eq.~(\ref{eq8b}) in the main text), we design their variants to incorporate the values of neighboring spins, which is defined as
\begin{equation}
\label{eqS1}
\begin{aligned}
    &m_{i} = \sum_{j\in N_a(i)} A_{ij}s_{j} h_{j},\\
    &\langle h_{i}\rangle_{MP} = \frac{m_{i}+h_i}{deg(i)+1},
\end{aligned}
\end{equation}
and
\begin{equation}
\label{eqS2}
\begin{aligned}
    &m_{i} = \sum_{j\in N_a(i)} J_{ij}s_{j} h_{j}, \\& \langle h_{i} \rangle _{MP} = \frac{m_i}{\sum_{j\in N_a(i)}|J_{ij}|} + h_{i},
\end{aligned}
\end{equation}
which are named Spin MP in GCon-VAN (we denote the combination as Spin GCon-VAN) and graph with spin message passing mechanism (Graph Spin MP), respectively. We then conduct the same numerical experiments as in Fig.~\ref{fig2} and Fig.~\ref{figS1}.

\begin{figure}[h]
\includegraphics[scale=0.6]{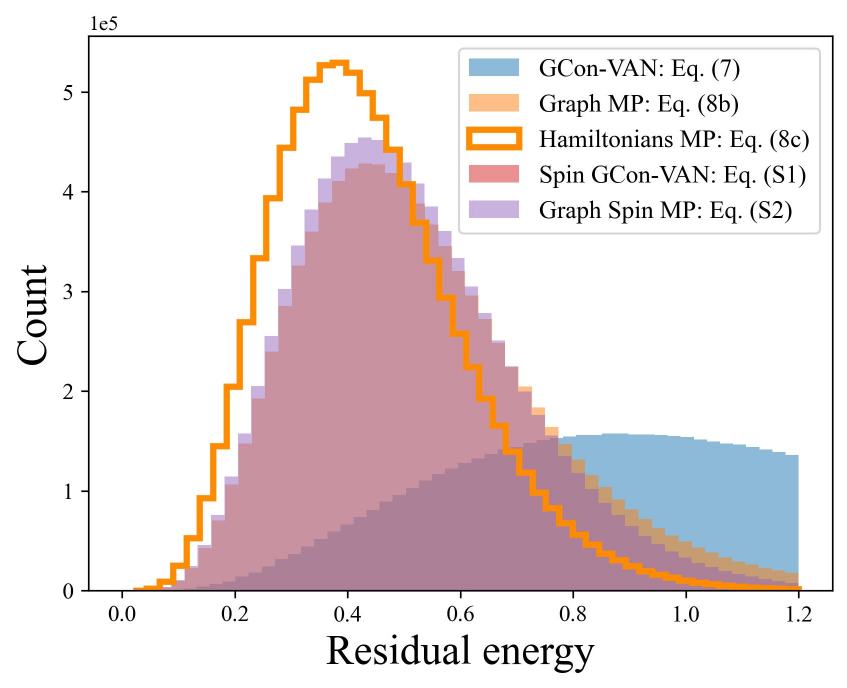}
\caption{\label{figS2} For the five MP mechanisms, the residual energy histogram on the WPE with system size $N=60$ and $\alpha=0.2$ same as in Fig.~\ref{fig1}. Each method contains $9\times10^{6}$ configurations obtained from 30 instances and each for 30 runs. Note that the results of Spin GCon-VAN are too poor to be shown in the figure.}
\end{figure}

It can be found in Fig.~\ref{figS2} that when considering the values of neighboring spins, the performance of Graph MP has slightly improved, while the performance of GCon-VAN has significantly deteriorated. Meanwhile, the performance of Graph Spin MP is still inferior to Hamiltonians MP when combined with VAN.

We now show the mean and standard deviation in Fig.~\ref{fig1}, Fig.~\ref{figS1} and Fig.~\ref{figS2} of the energy of $9\times10^6$ configurations from all the mechanisms mentioned above when they are combined with VAN.

\begin{table*}[h]
\caption{\label{tableS1} The mean and standard deviation in Fig.~\ref{fig1}, Fig.~\ref{figS1} and Fig.~\ref{figS2} of the energy of $9\times10^6$ configurations drawn from each method.}
\begin{ruledtabular}
\begin{tabular}{lcccccccc}
 mechanisms & VAN & VCA & GCon-VAN & MPVAN\footnote{Here MPVAN denotes the combination of VAN with Hamiltonians MP.} & Graph MP & Weighted MP & Spin GCon-VAN & Graph Spin MP \\ \hline
 $mean$ & 0.64064516 & 0.59639586 & 1.23419866 & \textbf{0.43949546} & 0.53990054 & 0.63682414 & 5.95293475 & 0.51100764 \\
 $std$ & 0.28714161 & 0.24911534 & 0.63625889 & \textbf{0.17004693} & 0.23393397 & 0.27861996 & 2.44381817 & 0.20037946 \\
\end{tabular}
\end{ruledtabular}
\end{table*}

\begin{figure}[h] 
\includegraphics[scale=0.6]{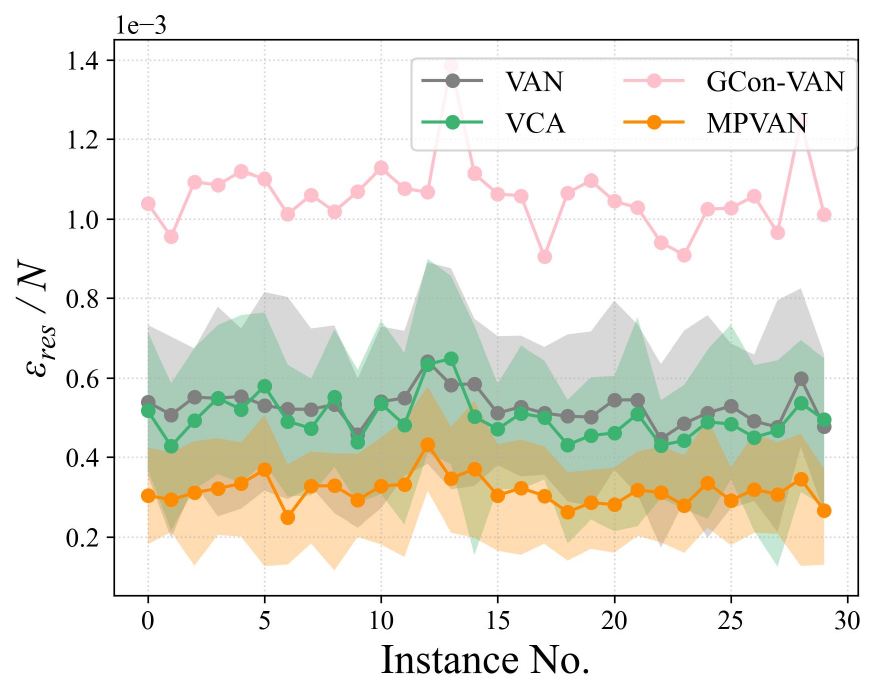}
\caption{\label{figS3} For MPVAN and existing methods, the residual energy per site on the 30 WPE instances with $N=60$ and $\alpha=0.2$ and each for 30 runs.}
\end{figure}

Third, we compare the residual energy per site for MPVAN with existing methods across 30 instances of the WPE in Fig.~\ref{figS3}. These problem instances are hard to solve, and none of the above methods can find the ground state. The results of GCon-VAN are poor and the area between the maximum and minimum values is too large, so we do not plot its corresponding color band.

As depicted in Fig.~\ref{figS3}, MPVAN consistently achieves the best results on all instances. Notably, the residual energy decreases by a substantial margin, ranging from 16.82\% to 33.49\%, when compared to the results obtained using other methods. Also, the differences between instances are large, which indicates the necessity of using the average of 30 instances to reflect the general properties of problems in other experiments.

It can also be seen that GCon-VAN yields the least favorable performance, which is a supplemental result to Ref.~\cite{Panfeng2021} when problems are defined on dense graphs.

Fourthly, since a order of spins plays a critical role in the variational conditional probability $q_\theta(s_i|s_1,s_2,\dots,s_{i-1})$ in Eq.~(\ref{eq2}) of the main text and autoregressive message passing process, we consider the impact of it within MPVAN. To investigate this, we select two instances from Fig.~\ref{figS3} where the advantages in $\epsilon_{res}$ of MPVAN over other methods are the smallest (the $7_{th}$ instance) and biggest (the $24_{th}$ instance).

\begin{figure}[h] 
\includegraphics[scale=0.5]{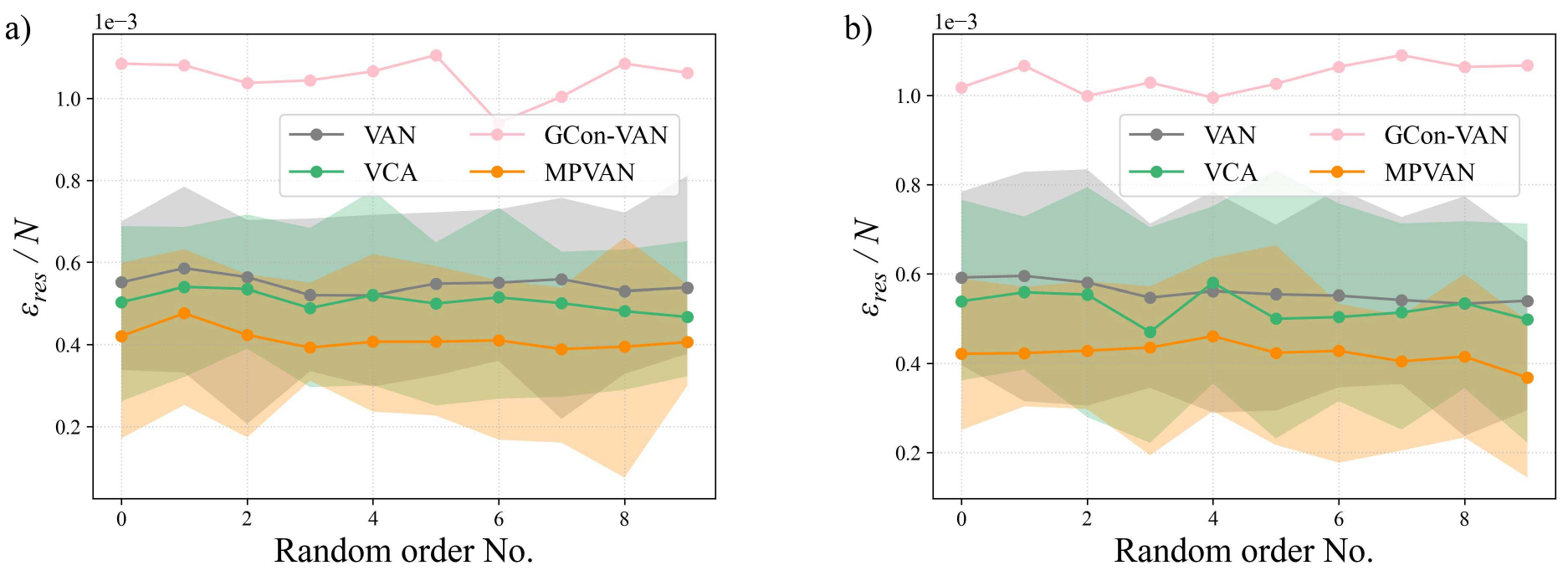}
\caption{\label{figS4} For MPVAN and existing methods with 10 random orders of spins, the residual energy per site. (a) On the $7_{th}$ instance. (b) On the $24_{th}$ instance in Fig.~\ref{figS3}, and each instance for 30 runs.}
\end{figure}

In these two instances, we randomly generate 10 orders of spins and evaluate the performance of MPVAN and existing methods. As illustrated in Fig.~\ref{figS4}, MPVAN consistently achieves the best results on all orders of spins. The results suggest that MPVAN may not exhibit a non-ignorable dependence on the order of spins, a characteristic similar to the standard autoregressive model \cite{Uria2016NADE}.

We show the $\epsilon_{res}$ of all mechanisms mentioned when combined with VAN in Tab.~\ref{tableS2} on the WPE with $N=60$ and $\alpha=0.2$. To demonstrate the differences among instances, we also show the standard deviation of residual energy on 30 instances as $std(\epsilon)$.

\begin{table*}[h]
\caption{\label{tableS2} The $\epsilon_{res}$ and $std(\epsilon)$ of all mechanisms combined with VAN on 30 WPE instances with $N=60$ and $\alpha=0.2$.}
\begin{ruledtabular}
\begin{tabular}{lcccccccc}
 mechanisms & VAN & VCA & GCon-VAN & MPVAN & Graph MP & Weighted MP &Spin GCon-VAN & Graph Spin MP \\ \hline
 $\epsilon_{res}$ & 0.03163008 & 0.02993632 & 0.06354465 & \textbf{0.01890272} & 0.02790873 & 0.03157805 & 0.28992669 & 0.02700393 \\
 $std(\epsilon)$ & 0.00245206 & 0.00322317 & 0.00547017 & \textbf{0.00214708} & 0.00227124 & 0.00265405 & 0.02453801 & 0.00218599 \\
\end{tabular}
\end{ruledtabular}
\end{table*}

It can be seen that MPVAN outperforms all other methods in mean($\epsilon_{res})$, and smaller $std(\epsilon_{res})$ denotes that it works stably on different instances.

\section{The optimal number of layers for MPVAN}
\label{appen2}

Formally, MPVAN and GNN are similar. Meanwhile, the number of layers is an important hyperparameter of GNN, which can cause excessive smoothing when the number of layers is too big, and the effect tends to perform poorly when the number of layers is too small as well. Therefore, we also consider the effect of the number of layers on the performance of MPVAN here.

\begin{figure}[h]
\includegraphics[scale=0.45]{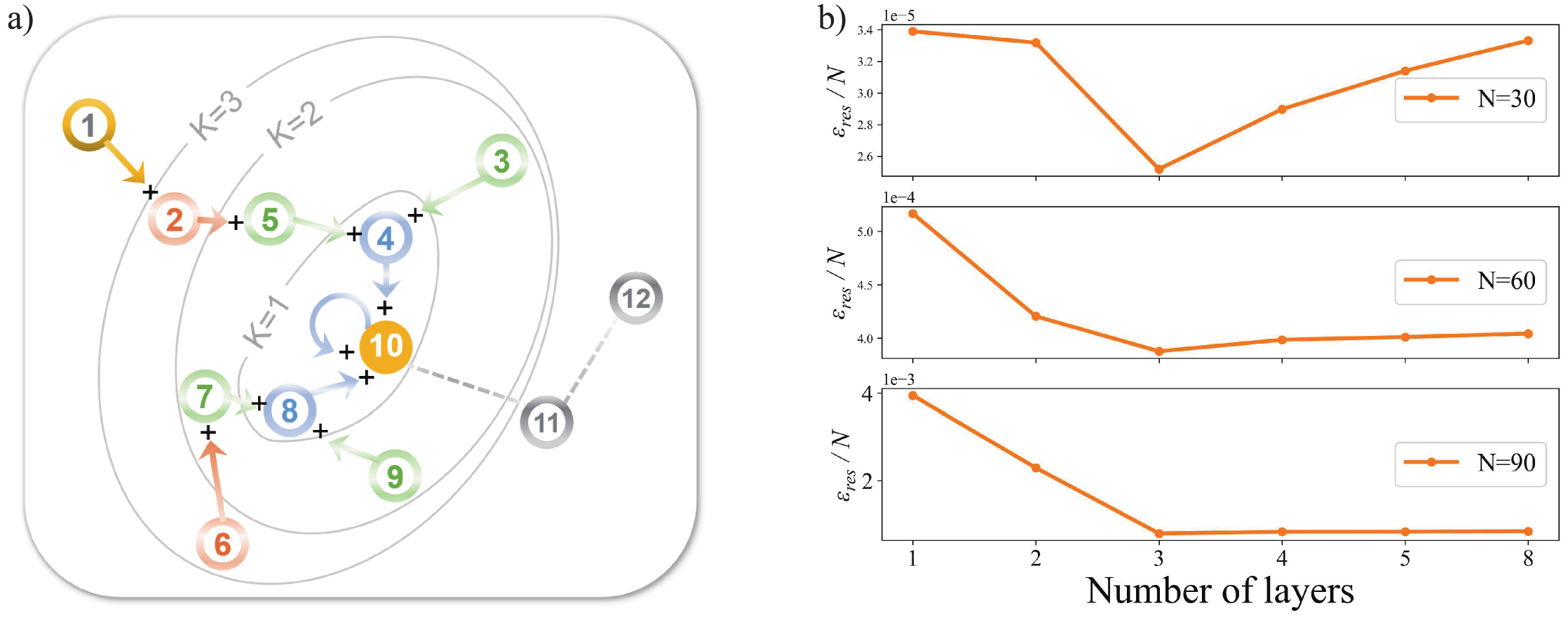}
\caption{\label{figS5} Schematic of the multi-layer autoregressive message passing process and experimental results for MPVAN with different numbers of layers. (a) The multi-layer autoregressive message passing process. For node 10, we highlight the 3-hops neighboring nodes it employs in message passing using different colors (the message passing is illustrated by the solid arrow, and the neighborhood is illustrated by the $k=1,2,3$ circle for node 10), and nodes 11 and 12 are not utilized to preserve the autoregressive property. (b) The residual energy per site varies with the number of layers for MPVAN across system sizes $N = 30, 60$, and 90 on the WPE with $\alpha=0.2$.}
\end{figure}

Message passing process with one layer utilizes features from one-hop neighboring nodes. Moreover, by stacking MPVAN layers, we can access neighboring features from nodes multiple hops away. As demonstrated in Fig.~\ref{figS5}(a), we visualize the message passing process of three layers for node 10. It is important to note that nodes 11 and 12 are deliberately excluded from the MP process to preserve the autoregressive property.

We conduct experiments to determine the optimal number of layers for MPVAN. As illustrated in Fig.~\ref{figS5}(b), the residual energy exhibits a trend of decreasing and then increasing as we vary the number of layers. Notably, the residual energy reaches its minimum when there are 3 layers, a result consistent across system sizes, including $N=30$, $60$, and $90$. Fewer layers limit the ability of networks to pass features from more distant neighboring nodes, while a large number of layers can lead to excessive smoothing of node features, making them indistinguishable. Therefore, we always adopt a 3-layer network for MPVAN.

\section{The occasionality in heuristic methods training} 
\label{appen3}

In the main text, we consistently employ the average of multiple runs for the same instance to provide a representative overview of its general p. Here we present the residual energy for 30 individual runs on a single WPE instance with $N=60$ and $\alpha=0.2$.

\begin{figure}[h]
\includegraphics[scale=0.6]{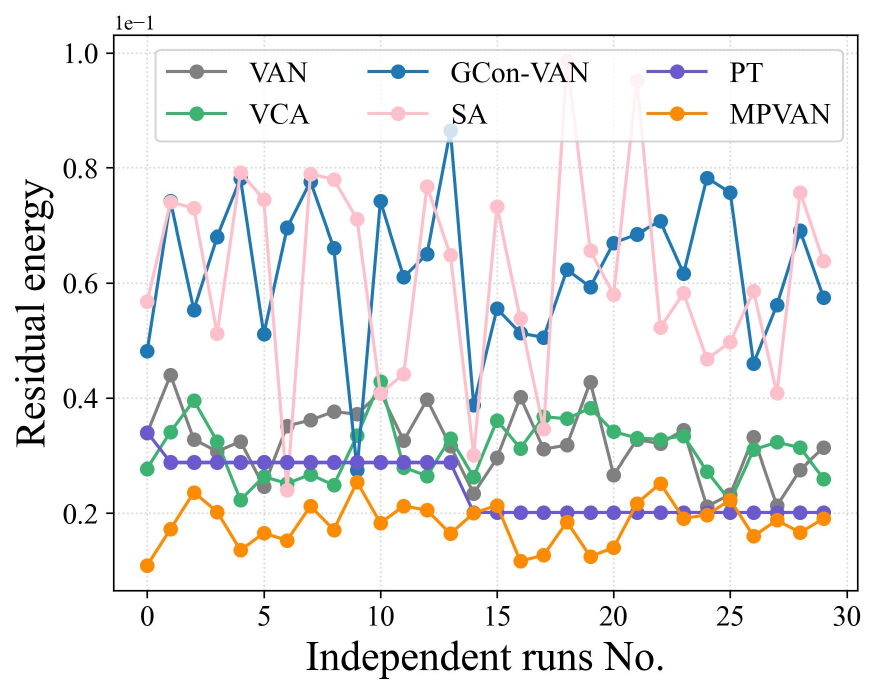}
\caption{\label{figS6} The residual energy of 30 independent runs on the $0\ th$ instance in Fig.~\ref{fig4} of the main text.}
\end{figure}

The residual energy of 30 independent runs are presented in Fig.~\ref{figS6}, which shows the occasionality in heuristic methods training. There is substantial divergence in the outcomes of 30 independent runs of the neural networks, with some runs yielding results that are as much as 168.94\% to 411.36\% higher than the minimum residual energy achieved by the respective methods. This significant variation underscores the necessity of employing the average of multiple runs to accurately reflect the general properties in other experiments.

\section{The approximating ratio results corresponding to Figure.~\ref{fig4}} 
\label{appen4}

We show the approximating ratio results corresponding to Fig.~\ref{fig4} of the main text in Fig.~\ref{figS7}. Similar to residual energy, MPVAN always provides a better approximation of the energy of the ground state than benchmark algorithms on the WPE and the SK model.

\begin{figure}[h]
\includegraphics[scale=0.55]{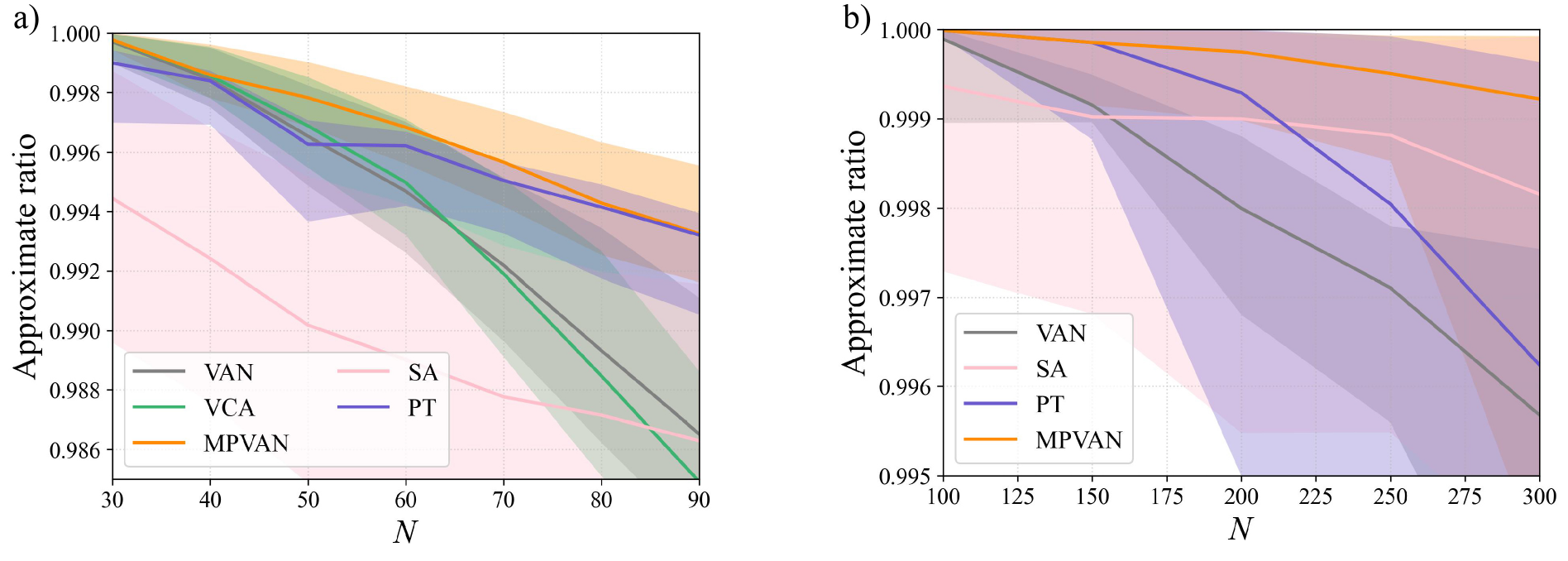}
\caption{\label{figS7} The approximating ratio results corresponding to Fig.~\ref{fig4} of the main text. (a) The results on the WPE. (b) The results on the SK model.}
\end{figure}

Notably, for the WPE with $N=30$ and $N=40$, there are some instances where some methods can find the ground state with a non-negligible probability. However, when $N\geq 50$, none of the methods can identify the ground state for any of the instances.

\section{The training speed} 
\label{appen5}

We also evaluate the training speeds of VCA, VAN, and MPVAN to provide a comprehensive understanding of their computational efficiency. The training speed of VCA is slower compared to VAN and MPVAN due to its reliance on the RNN structure. In the RNN, hidden units from the same layer must be computed sequentially, whereas VAN and MPVAN can be computed concurrently. To quantify these differences, we record the running speeds of MPVAN, VAN, and VCA, with the results summarized in Tab.~\ref{tableS2}. The data presented represents the time required for a single training step on a same instance, and all methods are evaluated on the same NVIDIA 3090 GPU and with the same hyperparameters.

\begin{table}[h]%
\caption{\label{tableS3}Time for 1 training step of MPVAN, VAN, and VCA.}
\begin{ruledtabular}
\begin{tabular}{lccc}
\textrm{method\footnote{Note that all the methods have the same hyperparameters and are compared on the same NVIDIA 3090 GPU.}}&
\textrm{N=30}&
\multicolumn{1}{c}{\textrm{N=60}}&
\textrm{N=90}\\
\colrule
MPVAN & 0.042 sec. & 0.068 sec. & 0.093 sec.\\
VAN & 0.052 sec. & 0.087 sec. & 0.132 sec.\\
VCA & 0.118 sec. & 0.486 sec. & 0.952 sec.\\
\end{tabular}
\end{ruledtabular}
\end{table}

As presented in Tab.~\ref{tableS2}, the training time of VCA for identical parameters is approximately $N/10$ times longer than that of VAN and MPVAN for instances with a system size of $N$. To maintain manageable computational requirements, we limit our comparisons to VCA on instances with $N\leq 100$ in the remaining experiments.

\section{The negative entropy at different learning rates} 
\label{appen6}

Here we show the negative entropy at different learning rates as a supplement to Fig.~\ref{fig3} of the main text.

\begin{figure}
\includegraphics[scale=0.18]{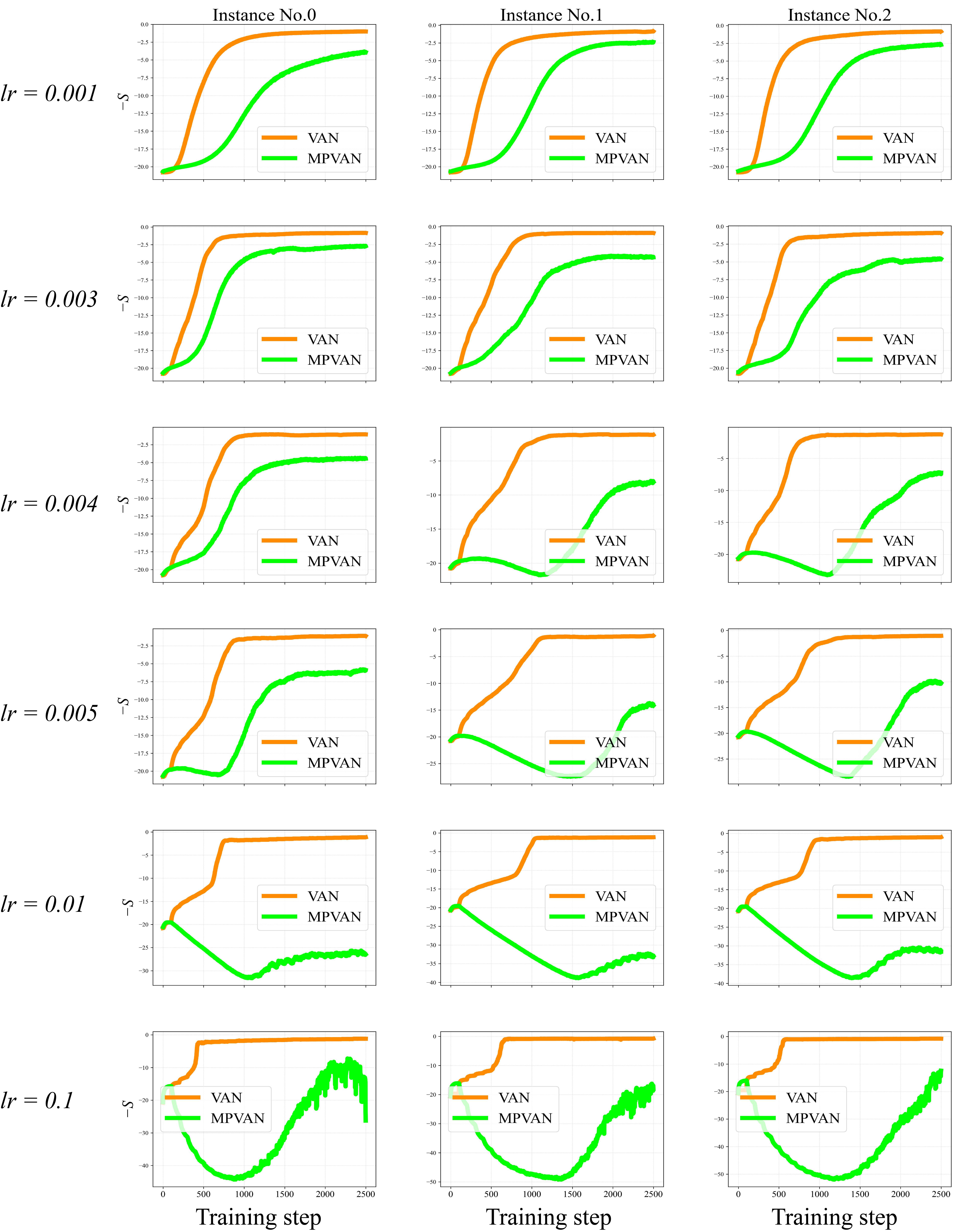}
\caption{\label{figS8} The negative entropy at different learning rates during training on three WPE instances with $N=30$ and $\alpha=0.2$. Note that for all learning rates, $N_{annealing} = 25$ and $N_{training}=100$ and run 10 times on each instance. Figure.~\ref{fig3} of the main text is the result with $lr=0.1$ on Instance No.0.}
\end{figure}

First, the distribution of MPVAN is consistently more uniform than that of VAN. When mode collapse emerges in VAN, MPVAN still has a larger negative entropy, indicating a more uniform distribution. Therefore, MPVAN delays the emergence of mode collapse. The negative entropy of MPVAN is smaller because the message passing changes the variational distribution, making some conditional probabilities small, resulting in a decrease in the joint probability of the configurations.

Second, when the learning rate is greater than 0.004, there is a decrease in negative entropy during training. It is because when learning rate is large, the training samples have a great impact on the back propagation of the neural network. The change gets stronger as the learning rate increases, making some conditional probabilities particularly small, so the negative entropy will be extremely small. Also, the minimum negative entropy decreases with increasing learning rates.

In summary, regardless of the choice of learning rates, the distribution of MPVAN is always more uniform than that of VAN, and mode collapse occurs at lower temperatures. Therefore, MPVAN delays the emergence of mode collapse.

\end{document}